\documentclass[journal]{IEEEtran}

\usepackage{color,xcolor}
\usepackage{array}
\usepackage{booktabs}
\usepackage{threeparttable}
\usepackage{multirow}
\usepackage{arydshln}

\usepackage{graphicx}
\usepackage{comment}
\usepackage{hyperref}
\hypersetup{hidelinks,
	colorlinks=true,
	allcolors=black,
	pdfstartview=Fit,
	breaklinks=true
}
\usepackage{amsmath,amssymb,amsfonts,amsthm,bm}
\usepackage{latexsym}
\usepackage[ruled,linesnumbered]{algorithm2e}
\usepackage{textcomp}
\usepackage[capitalize, nameinlink]{cleveref}
\newtheoremstyle{IEEEtheorem}
  {}
  {}
  {}
  {10pt}
  {\itshape}
  {:}
  { }
  {\thmname{#1}\thmnumber{ #2}\thmnote{ (#3)}}

\DeclareMathOperator*{\argmin}{arg\,min}
\theoremstyle{IEEEtheorem}

\newtheorem{thm}{Theorem}
\newtheorem{ass}{Assumption}

\newtheorem{lem}{Lemma}
\newtheorem{rem}{Remark}
\newtheorem{dfn}{Definition}
\newtheorem{pro}{Proposition}

\newtheorem{prob}{Problem}

\newcommand{{\R}}{{\mathbb{R}}}

\renewcommand{\baselinestretch}{1.0} 

\begin{document}
%
\title{Fixed-Time Convergent Distributed Observer Design of Linear Systems: A Kernel-Based Approach}
%
%
%

\author{Pudong~Ge,~\IEEEmembership{Student Member,~IEEE,}
Peng~Li,~\IEEEmembership{Member,~IEEE,}
         Boli~Chen,~\IEEEmembership{Member,~IEEE,}
         and~Fei~Teng,~\IEEEmembership{Senior~Member,~IEEE}
\thanks {This work was supported by EPSRC under Grant EP/W028662/1 and by The Royal Society under Grant RGS/R1/211256. P. Li is also supported by the Guangdong Basic and Applied Basic Research Foundation (2021A1515110262, 2022A1515011274).
(\textit{Corresponding author: Fei Teng}).}
\thanks{P. Ge and F. Teng are with the Dept. of Electrical and Electronic Engineering at Imperial College London, UK {\tt\small (pudong.ge19@imperial.ac.uk, f.teng@imperial.ac.uk)}}
\thanks{P. Li is with the School of Mechanical Engineering and Automation at Harbin Institute of Technology, Shenzhen, China {\tt\small (lipeng2020@hit.edu.cn)}}%
\thanks{B. Chen is with the Dept. of Electronic and Electrical Engineering at University College London, UK {\tt\small (boli.chen@ucl.ac.uk)}}
}

\markboth{IEEE Transactions on Automatic Control, Accepted}%
{IEEE Transactions on Automatic Control, Accepted}
%

\maketitle

\begin{abstract}
The robust distributed state estimation for a class of continuous-time linear time-invariant systems is achieved by a novel kernel-based distributed observer, which, for the first time, ensures fixed-time convergence properties. The communication network between the agents is prescribed by a directed graph in which each node involves a fixed-time convergent estimator. The local observer estimates and broadcasts the observable states among neighbours so that the full state vector can be recovered at each node and the estimation error reaches zero after a predefined fixed time in the absence of perturbation. This represents a new distributed estimation framework that enables faster convergence speed and further reduced information exchange compared to a conventional Luenberger-like approach. The ubiquitous time-varying communication delay across the network is suitably compensated by a prediction scheme. Moreover, the robustness of the algorithm in the presence of bounded measurement and process noise is characterised. Numerical simulations and comparisons demonstrate the effectiveness of the observer and its advantages over the existing methods. 
\end{abstract}


\section{Introduction}
\IEEEPARstart{L}{arge} scale systems are encountered more frequently in real-world applications, such as power networks, intelligent transportation systems and other cyber-physical systems. Such systems have an increasing demand for flexibility and scalability. The continuous growth of communication technology has enabled the development of decentralised and distributed solutions, which can perform collaborative tasks by using multi-agent communications. This has posed new challenges in control theory, including distributed consensus control, distribution estimation and so on \cite{liu2008controllability,Zhang:tcst19}. 

State estimation represents one of the most important problems in control. Motivated by previous developments in the centralised observer, this paper focuses on the distributed observer, where the outputs of a large scale system are measured by a sensor network and only a small portion of the system output is available at each sensor node. Therefore, the main challenge is that the state of the system is not fully observable at any sensor node. The goal is to design a distributed observer, such that the full state of the system can be collaboratively reconstructed by each agent using local measurement and proper neighbouring communication~\cite{zou2020moving,farina2010distributed,wang2020distributed,park2012necessary}.

A variety of distributed linear time-invariant (LTI) state estimation approaches have been reported in the literature under different formulations inherit from the centralised approaches, including the Kalman filter and Luenberger observer. A comprehensive overview of existing distributed observers can be found in  \cite{Knotek:tcst21}.  It is noteworthy that the design of a distributed observer is highly influenced by the communication graph. In \cite{olfati2007distributed}, a distributed Kalman filtering algorithm is proposed for an undirect and connected communication graph. With the same assumption on the communication graph, a distributed Luenberger-type observer is presented in \cite{kim2016distributed}.  More recently, research efforts are paid to more general directed graphs \cite{han2019simple,park2017design,Mitra:tac19,wang:tac18}. Necessary and sufficient observability conditions for designing a distributed observer are stated in \cite{park2017design,ugrinovskii2013conditions,khan2010connectivity}. {The study \cite{del2019distributed}, on top of the Luenberger observer based scheme, introduces a multi-hop staircase decomposition mechanism, which makes it possible to lower information exchange and to relax the common assumptions of strongly connected graphs compared to the majority of distributed observers in the literature.}
Most of the existing methods are based on the Luenberger observer, which permits a single-agent-based design and implementation and ensures the local state estimate of each agent asymptotically converges the system state. On the other hand, {\cite{silm2019note,silm2019design} propose alternative solutions to distributed state estimation by using the homogeneous technique \cite{Perruquetti:tac08}.} As such, the state estimation error decays within a small finite time. 

An important challenge in distributed estimation that has not been extensively addressed is the communication delay throughout the network. In the majority of existing works, the effect of delays is omitted, while its presence may drastically influence the estimation performance. Very recently, \cite{silm:tac20} proposes a time-delay distributed observer, which guarantees exponential stability in the presence of time-varying but conservatively known (upper bound is available) communication delays, and the convergence rate can be designed up to a maximum total delay. 

In this paper, we study the distributed observer problem of a continuous-time LTI system, where the communication between agents may involve time-varying delays, as assumed in \cite{silm:tac20}. The main contribution of the paper lies in a novel fixed-time convergent distributed observer based on a cross-agent information sharing mechanism. The method provides an example of how distributed estimation systems can benefit from fixed-time convergence properties. 
The key to the fixed-time observer is the Volterra integral operators with specialised kernel functions, as inspired by the centralised counterpart \cite{pin2013kernel}. {In contrast to the majority of methods in the literature that require the full-dimensional state estimates to be shared among neighbouring nodes, the proposed scheme enables a reduction of the transmitted data over the communication links by invoking a rank-condition and the effect of delays in communication networks is compensated.}
Finally, the robustness of the proposed method against measurement noise and perturbations is characterised. 

The outline of this paper is as follows. The state estimation problem formulation and mathematical preliminaries are given in \cref{sec:2}. \cref{sec:3} introduces the main algorithm, and its robustness against disturbances and measurement noise is analysed in \cref{sec:4}. Simulation examples are presented in \cref{sec:5}, and concluding remarks and future work are discussed in \cref{sec:6}.

\section{Problem Statement and Preliminary}
\label{sec:2}
\subsection{Problem Setting}
\textit{Notation:} Let $\R$, ${\R}_{\geq 0}$  and ${\R}_{> 0}$  denote the real, the non-negative real and the strict positive real sets of numbers, respectively. Given a vector ${\bm x} \in \R^n$, we denote $|{\bm{x}}|$ as the Euclidean norm of ${\bm x}$. Given a time-varying vector ${\bm x}(t) \in \R^n$, $t \in \R_{\geq 0}$, we will denote $\left\|{\bm x}\right\|_{\infty}$ as the quantity $\left\|{\bm x}\right\|_\infty = \sup{}_{t\geq0}|{\bm x}(t)|$. Assuming ${\bm x}(t)$ is $k$-th order differentiable, the $k$-th order derivative signal of $\bm{x}(t)$ is denoted by ${\bm x}^{(k)}(t)$. 

In this paper, a directed graph is denoted by $\mathcal{G}=\{\mathcal{N},\mathcal{E},\mathcal{A}\}$, where $\mathcal{N}=\{1,2,\cdots,N\}$ is a set of nodes, $\mathcal{E} \subseteq \mathcal{N} \times \mathcal{N}$ is a set of edges, and $\mathcal{A}=[a_{ij}]\in \mathbb{R}^{N \times N}$ denotes the adjacency matrix. The element $a_{ij}$ is the weight of the edge $(i,j)$, and $a_{ij} = 1$ if and only if $(i,j)\in\mathcal{E}$ and $a_{ij}=0$ otherwise. Specifically, $(i,j)\in\mathcal{E}$ means that the $i$-th node can send information to the $j$-th node. The set of neighbours of node $j$ is described by $\mathcal{N}_{j}=\{i:(i,j)\in \mathcal{E}\}$. A graph $\mathcal{G}$ is strongly connected if there exists a directed path between $\forall i,j\in\mathcal{N}, i\neq j$. Given a set $\{\bm{G}_1,\bm{G}_2,\cdots,\bm{G}_N\}$ of matrices with ${\bm G}_i \in \R^{m\times n}$, we use $\mathrm{col}(\bm{G}_{1},\bm{G}_{2},\cdots,\bm{G}_{N})$ to denote the stacked matrix $[\bm{G}_{1}^\top,\bm{G}_{2}^\top,\cdots,\bm{G}_{N}^\top]^\top\in \mathbb{R}^{N m\times n}$ and $\mathrm{diag}(\bm{G}_{1},\bm{G}_{2},\cdots,\bm{G}_{N})\in \mathbb{R}^{N m\times N n}$ to denote the block diagonal matrix with the $\bm{G}$’s along the diagonal. The following definitions will also be used $ \mathrm{col}_{i\in\mathcal{N}}(\bm{G}_{i})\triangleq \mathrm{col}(\bm{G}_{1},\bm{G}_{2},\cdots,\bm{G}_{N}) $ and $\mathrm{diag}_{i\in\mathcal{N}}(\bm{G}_{i})\triangleq \mathrm{diag}(\bm{G}_{1},\bm{G}_{2},\cdots,\bm{G}_{N})$. $|\mathcal{N}|$ defines the cardinality of the set. 
$\mathrm{obsv}(\cdot,\cdot)$ and $\mathrm{rank}(\cdot)$ are used to define 
the observability matrix of the given system and matrix rank, respectively. 

Consider the following continuous LTI system
\begin{align}
	\bm{\dot{x}} = \bm{Ax},\ \bm{y} = \bm{Cx}
	\label{eq:sys}
\end{align}
where $\bm{x}\in\mathbb{R}^{n}$ is the state and $\bm{y}\in\mathbb{R}^{m}$ is the output, $\bm{A}\in\mathbb{R}^{n\times n},\bm{C}\in\mathbb{R}^{m\times n}$. The system \eqref{eq:sys} is sensed by $N$ distributed agents 
$
{ \bm y}_i = \bm{C}_i{\bm x}
$
with $\bm{y}=\mathrm{col}(\bm{y}_{1},\bm{y}_{2},\cdots,\bm{y}_{N})$ where $\bm{y}_{i}\in\mathbb{R}^{m_{i}}$, $\sum_{i=1}^{N}m_{i}=m$ and $\bm{C}=\mathrm{col}(\bm{C}_{1},\bm{C}_{2},\cdots,\bm{C}_{N})$. For each node/subsystem $i\in\mathcal{N}$, $\bm{y}_{i}$ is the only output that is available for node $i$.
Neighbour relations between distinct pairs of agents are characterised by a directed graph $\mathcal{G}$. We assume throughout that $\bm{C}_{i} \neq 0, \forall i \in \mathcal{N}$ and $\bm{C}_{i}\neq \bm{C}_{j}, \,i\neq j, \,i,j \in \mathcal{N}$. For the sake of further analysis, let $\mathcal{O}\triangleq \mathrm{obsv}(\bm{A},\bm{C})$ and $\mathcal{O}_i \triangleq \mathrm{obsv}(\bm{A},\bm{C}_i)$ be the observability matrices of the pair $(\bm{A},\bm{C})$ and $(\bm{A},\bm{C}_i)$, respectively. The date transmission between agents may be impact by time-varying delays.

\begin{ass}\label{ass:observability}
    The pair $(\bm{A},\bm{C})$ is observable, but the pair $(\bm{A},\bm{C}_{i})$ is not fully observable.
\end{ass}


The problem investigated in this article is defined as follows.
\begin{prob}\label{prob:1}
    Given the system  \eqref{eq:sys} subject to a communication topology $\mathcal{G}$, how to design a distributed observer with the estimated state $\hat{\bm x}_i,\,\forall i \in \mathcal{N}$, such that the estimation error goes to 0 within a fixed time, 
    \begin{equation}\label{rq:mainproblem}
        |\hat{\bm x}_i(t) - \bm{x}(t)|=0,\, \forall t \geq \overline{\tau}
    \end{equation}
    where $\overline{\tau} \in \R_{>0}$ is a known finite time.
\end{prob}

\subsection{The Volterra Operator and BF-NK}\label{subsec:volterra}
Volterra operator and non-asymptotic kernel functions are the key tools to the observer design in the paper. To introduce later a distributed fixed-time observer,  here we briefly recall the basic concepts~\cite{pin2013kernel}. 

Given a function belongs to the Hilbert space of locally integrable function with domain $\mathbb{R}_{\geq0}$ and range $\mathbb{R}$, i.e.,  $w\in\mathcal{L}_{loc}^{2}(\mathbb{R}_{\geq0})$, its image by the \textit{Volterra operator} $V_{K}$ induced by a Hilbert-Schmidt ($\mathcal{HS}$) Kernel Function $K: \mathbb{R}\times\mathbb{R}\rightarrow\mathbb{R}$ is denoted by $\left[V_{K}w\right]$ of the form
\begin{align*}
	\left[V_{K}w\right](t)\triangleq\int_{0}^{t}K(t,\tau)w(\tau)d\tau,\quad t\in\mathbb{R}_{\geq0}
\end{align*}

\begin{dfn}[BF-NK]
	\cite{pin2013kernel} If a kernel $K\in\mathcal{HS}$ which is at least $(i-1)$th order differentiable with respect to the second argument, verifies the conditions
$
		 K^{(j)}(t,0)=0, \, \forall t\in \mathbb{R}_{\geq 0}
		 $ and $
		 K^{(j)}(t,t)\ne 0, \, \forall t\ne 0
$
	for all $j\in\{0,1,\cdots,i-1\}$, it is called an $i$th order \textit{Bivariate Feedthrough Non-asymptotic Kernel}~(BF-NK).
\end{dfn}

\begin{lem}
	\cite{pin2016non} For a given $i\geq0$, consider a signal defined as a function of time $w\in\mathcal{L}^{2}(\mathbb{R}_{\geq0})$ that admits the $i$th derivative in $\mathbb{R}_{\geq0}$ and a kernel function $K\in\mathcal{HS}$, having the $i$th derivative with respect to the second argument, denoted as $K$. After successive integral by parts, it holds that
	\begin{multline}
			\left[V_{K}w^{(i)}\right]\! (t)=\sum_{j=0}^{i-1}(-1)^{i-j-1}w^{(j)}(t)K^{(i-j-1)}(t,t)  \\
			+\sum_{j=0}^{i-1}(-1)^{i-j}w^{(j)}(0)K^{(i-j-1)}(t,0) 
			+(-1)^{i}\left[V_{K^{(i)}}w\right]\!(t)
		\label{eq:lem_volterra}
	\end{multline}
	which shows the function $\left[V_{K}w^{(i)}\right](t)$ is non-anticipative with respect to the lower-order derivatives $w,w^{(1)},\cdots,w^{(i-1)}$.
\end{lem}
Owing to the definition of the BF-NK, induced by a BF-NK $K_{h}$, the Volterra image \eqref{eq:lem_volterra} reduces to
\begin{align}
	\begin{aligned}
		\left[V_{K}w^{(i)}\right] (t)&=\sum_{j=0}^{i-1}(-1)^{i-j-1}w^{(j)}(t)K^{(i-j-1)}(t,t) \\
		&+(-1)^{i}\left[V_{K^{(i)}}w\right](t)
	\end{aligned}
	\label{eq:BF-NK_volterra}
\end{align}
A typical class of $\delta$th order BF-NKs that we will use in this paper have the form of
\begin{align*}
	K_{h}(t,\tau)=e^{-\omega_{h}(t-\tau)}\left(1-e^{-\overline{\omega}\tau}\right)^{\delta}
\end{align*}
which is parameterised by $\omega_{h}\in\mathbb{R}_{>0}$ and $\overline{\omega}\in\mathbb{R}_{>0}$. As it can be seen, all the non-asymptoticity conditions up to the $\delta$-th order are met thanks to the factor $\left(1-e^{-\bar{\omega} \tau}\right)^{\delta}$ regardless of the choice of $\omega_{h}$ and $\bar{\omega}$. 

The Volterra image signal $\left[\mathcal{V}_{K_h^{(i)}} w\right](t), \forall\,i \in\{1,2,\cdots,\delta\}$ can be obtained as the output of a linear time-varying scalar system. Letting $\xi(t) = \left[\mathcal{V}_{K_h^{(i)}} w\right](t)$, we have that 
\begin{equation}
\begin{array}{lll}
\dot{\xi}(t)\!\!\!\!&=&\!\!\!\! K_h^{(i)}(t,t)w(t)+\displaystyle \int_0^t\left( \frac{\partial}{\partial t}K_h^{(i)}(t,\tau) \right)w(\tau)d\tau\\
\!\!\!\!&=&\!\!\!\! -\omega_h \xi(t)+K_h^{(i)}(t,t)w(t) 
\end{array}
\label{eq:xiltv2}
\end{equation}
with $\xi(0)=0$.  Being $K_h^{(i)}(t,t)$ bounded and $\omega$ strictly positive, it holds that the scalar dynamical system realization of the Volterra operators induced by the proposed kernels is BIBO stable with respect to $w$. 

\section{Fixed-Time Convergent Distributed Observer}
\label{sec:3}
In this section, the solution method to Problem~\ref{prob:1} is presented. In the first instance, data transmission and communication delays within the sensor network are omitted. Under such a condition, a new distributed observer framework is designed. Then, the algorithm is modified to accommodate various network delays.

\subsection{Delay-Free Case}
\label{subsec:withoutdelay}
From \cref{ass:observability}, the state vector $\bm{x}$ is not fully observable from a single sensor node. Nevertheless, by resorting to the {commonly used observability decomposition technique of each subsystem \cite{kim2016distributed,han2019simple,silm2019design}}, it is possible to partially estimate $\bm{x}$. 
Let $n_{i}$ denotes the rank of the observability matrix of $(\bm{A},\bm{C}_{i})$, that is $n_{i}\triangleq \mathrm{rank}(\mathcal{O}_i)<n$.
There exists an orthogonal matrix $\bm{T}_i \in \mathbb{R}^{n \times n}$ that enables the state transformation, $\bm{\bar{x}}_{i}=\bm{T}_i\bm{x}$, and $\bm{\bar{x}}_{i}$ admits the following decomposition $\bm{\bar{x}}_{i}=
\begin{bmatrix}
\bm{\bar{x}}_{io}\quad\bm{\bar{x}}_{iu}
\end{bmatrix}^\top=
\begin{bmatrix}
\bm{T}_{io}\quad\bm{T}_{iu}
\end{bmatrix}^\top\bm{x}
$,
where $\bm{\bar{x}}_{io}$ represents the observable part and $\bm{\bar{x}}_{iu}$ stands for the unobservable part. The dynamics of $\bm{\bar{x}}_{i}$ follows
$
			 \bm{\dot{\bar{x}}}_{i}=\bm{\bar{A}}_i\bm{\bar{x}}_{i},\, 
		 \bm{y}_{i}=\bm{\bar{C}}_{i}\bm{\bar{x}}_{i},
$
where
\begin{align}
	\begin{aligned}
		\bm{\bar{A}}_i\!=\!\bm{T}_i\bm{A}\bm{T}_i^\top\!=\!\begin{bmatrix}
			\bm{A}_{io} & \bm{0}\\
			\bm{A}_{ir} & \bm{A}_{iu}
		\end{bmatrix},\,
		\bm{\bar{C}}_{i}\!=\!\bm{C}_{i}\bm{T}_i^\top\!=\!\begin{bmatrix}
			\bm{C}_{io} & \bm{0}
		\end{bmatrix}
	\end{aligned}\label{eq:decom_matrix}
\end{align}
with $\bm{A}_{io}\in\mathbb{R}^{n_{i}\times n_{i}}$, $\bm{A}_{iu}\in\mathbb{R}^{(n-n_{i})\times (n-n_{i})}$, $\bm{A}_{ir}\in\mathbb{R}^{(n-n_{i})\times n_{i}}$, $\bm{C}_{io}\in\mathbb{R}^{m_{i}\times n_{i}}$, $,\ \bm{T}_{io}\in\mathbb{R}^{n_{i}\times n},\ \bm{T}_{iu}\in\mathbb{R}^{(n-n_{i})\times n}$. Furthermore, $(\bm{A}_{io},\bm{C}_{io})$ is observable, 
and the dynamics of the observer part is governed by 
\begin{align}
    \bm{\dot{\bar{x}}}_{io} = \bm{A}_{io}\bm{\bar{x}}_{io},\ \bm{y}_{i} = \bm{C}_{io}\bm{\bar{x}}_{io}\,.
    \label{eq:obs_dynamic}
\end{align}
Next, a finite and fixed time convergent observer~\cite{pin2013kernel} is applied to estimate the observable part $\bm{\bar{x}}_{io}\in\mathbb{R}^{n_i}$, which will then be used to recover the full state vector through communication.

Thanks to the observability of $(\bm{A}_{io},\bm{C}_{io})$, there exists a linear coordinates transformation $\bm{z}_{i}=\bm{T}_{iz}\bm{\bar{x}}_{io}$ with $\bm{T}_{iz}\in\mathbb{R}^{n_{i}\times n_{i}}$ such that the system \eqref{eq:obs_dynamic} can be rewritten in the observer canonical form with respect to $\bm{z}_{i}$
\begin{align}
	\bm{\dot z}_{i}=\bm{A}_{i,z}\bm{z}_{i},\ y_{i}=\bm{C}_{i,z}\bm{z}_{i}
	\label{eq:obs_canonical}
\end{align}
where 
\begin{align*}
    \bm{A}_{i,z}&=\bm{T}_{iz}\bm{A}_{io}\bm{T}_{iz}^{-1}=\begin{bmatrix}
		a_{n_{i}-1} & 1 & \cdots & 0 \\
		\vdots & \vdots & \ddots  & \vdots \\
		a_{1} & 0 & \cdots  & 1 \\
		a_{0} & 0 & \cdots  & 0
	\end{bmatrix}
\end{align*}
For simplicity, we herein assume \eqref{eq:obs_dynamic} to be a single-output system, thereby $\bm{C}_{i,z}=\begin{bmatrix}
	1 & 0 & \cdots &0
\end{bmatrix}$. {However, the method is not limited to single-output systems as the state vector of a multi-output observable system can be estimated by multiple observers individually designed for each single output utilising, for example, the technique described in Lemma 9.4.4 of \cite{willems1998introduction}.}
The canonical-form subsystem \eqref{eq:obs_canonical} admits the following input-output realization
\begin{align}
	y_{i}^{(n_{i})}={\sum_{p=0}^{n_{i}-1}}{a_{p}y_{i}^{(p)}}
	\label{eq:io_dynamic}
\end{align}
Let us consider the Volterra integral operator induced by $K_i=e^{-\omega_{i,h}(t-\tau)}\left(1-e^{-\bar{\omega}_i\tau}\right)^{n_i}$, an $n_{i}$-th order BF-NK. Applying the Volterra integral operator introduced in \cref{subsec:volterra} and recalling \eqref{eq:BF-NK_volterra} for \eqref{eq:io_dynamic}, we obtain
\begin{align}
		&\sum_{p=0}^{n_{i}-1}\!(-1)^{n_{i}-p-1}y_{i}^{(p)}K_{i}^{(n_{i}-p-1)}(t,t) 
		\!+\!(-1)^{n_{i}}\!\left[V_{K_{i}^{(n_{i})}}y_{i}\right]\!(t)=
		\nonumber \\
		&\sum_{q=0}^{n_{i}-1}a_{q}\Bigg(\!(-1)^{q}\!\left[V_{K_{i}^{(q)}}y_{i}\right]\!(t)
		\!+\!\sum_{p=0}^{q-1}(-1)^{p+q-1}y_{i}^{p}K_{i}^{q-p-1}(t,t)\!\Bigg)
	\label{eq:volterra_eq1}
\end{align}
Then, for all $r\in\{0,\cdots,n_{i}-1\}$, the $r$-th state variable of \eqref{eq:obs_canonical} has the form of $z_{i,r}=y_{i}^{(r)}-\sum_{p=0}^{r-1}a_{n_{i}-r+p}y_i^{(p)}$, in terms of which
we rearrange \eqref{eq:volterra_eq1} after cumbersome algebra
\begin{align}
	\lambda_{i}=\bm{\gamma}_{i}\bm{z}_{i}
	\label{eq:obs_eq1}
\end{align}
\begin{align*}
	&\lambda_{i}\triangleq(-1)^{n_{i}-1}\left[V_{K_{i}^{(n_{i})}}y_{i}\right]+\sum_{p=0}^{n_{i}-1}a_{p}(-1)^{p}\left[V_{K_{i}^{(p)}}y_{i}\right]\\
	&\bm{\gamma}_{i}\triangleq\left[(-1)^{n_{i}-1}K_{i}^{(n_{i}-1)}(t,t),\cdots,K_{i}(t,t)\right]
\end{align*}
\cref{eq:obs_eq1} cannot be solved directly due to the rank-deficiency. However, by using $n_{i}$ BF-NKs $K_{i,h}(t,\tau)$ with different $\omega_{i,h},\,h\in\{0,\cdots,n_{i}-1\}$ but identical $\bar{\omega}_i$, it is possible to augment \eqref{eq:obs_eq1} into a matrix form
\begin{align}
	\bm{\Lambda}_i=\bm{\Gamma}_i\bm{z}_{i}
	\label{eq:obs_eq2}
\end{align}
where $\bm{\Lambda}_i=\left[\lambda_{i,0},\lambda_{i,1},\cdots,\lambda_{i,n_{i}-1}\right]^\top$ and $\bm{\Gamma}_i=\left[\bm{\gamma}_{i,0}^\top,\bm{\gamma}_{i,1}^\top,\cdots,\bm{\gamma}_{i,n_{i}-1}^\top\right]^\top$. In addition, all transformed signals $\left[V_{K_{i,h}^{(p)}}y_{i}\right],\forall h\in\{0,\cdots,n_{i}-1\},\forall p\in\{0,\cdots,n_{i}\}$ can be computed by \eqref{eq:xiltv2}. 
The invertibility of $\bm{\Gamma}$ is guaranteed for all $t>0$ thanks to the properties of the BF-NK~\cite{pin2019robust,li2020kernel}. Therefore, the observable state vector is estimated by \eqref{eq:obs_canonical}:
\begin{align}\label{eq:obs_eq20}
	\bm{\hat{z}}_{i}=\bm{\Gamma}_i^{-1}\bm{\Lambda}_i,\forall t \geq t_\delta
\end{align}
where $t_\delta$ is a small time instant, provided that $\bm{\Gamma}_i$ is invertible for any $t\geq t_\delta$, so as to circumvent the singularity $\bm{\Gamma}_i(0)=0$ as $K_{i,h}^{(h-1)}(0,0)=0$. Note that, in the proposed distributed observer, $t_\delta$ is set uniformly across all agents whereas the kernel functions for each agent are independently designed.

From the coordinate transformations $\bm{T}_{i\alpha}\triangleq \bm{T}_{iz}\bm{T}_{io} \in\mathbb{R}^{n_{i}\times n}$, which is known, it is immediate to obtain:
\begin{align}\label{eq:zi2x}
    \bm{z}_{i}=\bm{T}_{i\alpha}\bm{x},\,\,\forall i \in \mathcal{N}
\end{align}
As $\bm{T}_{i\alpha}$ is not invertible, the global state vector $\bm{x}$ can not be estimated from local state estimate $\bm{\hat z}$ via \eqref{eq:zi2x} without further information.
To establish the intercommunication requirements, the following results are shown.
\begin{lem}[Sylvester inequality]
	\cite{horn2012matrix} Let $\bm{G}\in\mathbb{R}^{m\times n}$ and $\bm{H}\in\mathbb{R}^{n\times p}$, it holds that 
	$
	\mathrm{rank}(\bm{G})+\mathrm{rank}(\bm{H})-n\leq\mathrm{rank}(\bm{GH})\leq\min\{\mathrm{rank}(\bm{G}),\mathrm{rank}(\bm{H})\}.
    $
	\label{lem:sylvester_inequ}
\end{lem}
\begin{pro}
	\label{pro:1}
  For any subset $\mathcal{N}_{s}\subset\mathcal{N}$, the observable matrix $\mathrm{obsv}\left(\bm{A},\mathrm{col}_{i\in\mathcal{N}_{s}}(\bm{C}_{i})\right)$ determined by $\left(\bm{A},\mathrm{col}_{i\in\mathcal{N}_{s}}(\bm{C}_{i})\right)$ satisfies the following condition: 
  \begin{align*}
      \mathrm{rank}\big(\mathrm{obsv}\left(\bm{A},\mathrm{col}_{i\in\mathcal{N}_{s}}(\bm{C}_{i})\right)\big)=\mathrm{rank}\big(\mathrm{col}_{i\in\mathcal{N}_{s}}(\bm{T}_{i\alpha})\big).
  \end{align*}
\end{pro}
\begin{proof}
	From the definition of $\bm{T}_{i\alpha}$, we have $\mathrm{col}_{i\in\mathcal{N}_{s}}(\bm{T}_{i\alpha})=\mathrm{diag}_{i\in\mathcal{N}_{s}}(\bm{T}_{iz})\mathrm{col}_{i\in\mathcal{N}_{s}}(\bm{T}_{io}).$ Owing to the decomposition \eqref{eq:decom_matrix},
	\begin{align*}
		\mathrm{rank}[\mathrm{col}_{i\in\mathcal{N}_{s}}(\bm{T}_{io})]=\mathrm{rank}[\mathrm{obsv}(\bm{A},\mathrm{col}_{i\in\mathcal{N}_{s}}(\bm{C}_{i}))]
	\end{align*}
	From \cref{lem:sylvester_inequ} and the fact that $\bm{T}_{iz}$ is full rank, we have
	\begin{align*}
		&\mathrm{rank}[\mathrm{diag}_{i\in\mathcal{N}_{s}}(\bm{T}_{iz})]+\mathrm{rank}[\mathrm{col}_{i\in\mathcal{N}_{s}}(\bm{T}_{io})]-\sum_{i\in\mathcal{N}_{s}}n_{i} \\
		&\leq\mathrm{rank}[\mathrm{col}_{i\in\mathcal{N}_{s}}(\bm{T}_{i\alpha})]=\mathrm{rank}[\mathrm{diag}_{i\in\mathcal{N}_{s}}(\bm{T}_{iz})\mathrm{col}_{i\in\mathcal{N}_{s}}(\bm{T}_{io})] \\
		&\leq\min\{\mathrm{rank}[\mathrm{diag}_{i\in\mathcal{N}_{s}}(\bm{T}_{iz})],\mathrm{rank}[\mathrm{col}_{i\in\mathcal{N}_{s}}(\bm{T}_{io})]\} \\
		&\hspace{4.5cm}\Downarrow\ (a)\\
		&\mathrm{rank}[\mathrm{col}_{i\in\mathcal{N}_{s}}\!(\bm{T}_{io})]\!\leq\!\mathrm{rank}[\mathrm{col}_{i\in\mathcal{N}_{s}}\!(\bm{T}_{i\alpha})]\!\leq\!\mathrm{rank}[\mathrm{col}_{i\in\mathcal{N}_{s}}\!(\bm{T}_{io})]\\
		&\hspace{4.5cm}\Downarrow \\
		&\mathrm{rank}[\mathrm{col}_{i\in\mathcal{N}_{s}}(\bm{T}_{i\alpha})]
		=\mathrm{rank}[\mathrm{col}_{i\in\mathcal{N}_{s}}(\bm{T}_{io})]\\
		&\hspace{4cm}=\mathrm{rank}[\mathrm{obsv}(\bm{A},\mathrm{col}_{i\in\mathcal{N}_{s}}(\bm{C}_{i}))]
	\end{align*}
	where $(a)$ comes from the fact that
	\begin{align*}
		\max_{i\in\mathcal{N}_{s}}n_{i}\!\leq\!\mathrm{rank}[\mathrm{col}_{i\in\mathcal{N}_{s}}(\bm{T}_{io})]\!\leq\!\mathrm{rank}[\mathrm{diag}_{i\in\mathcal{N}_{s}}(\bm{T}_{iz})]\!=\!\sum_{i\in\mathcal{N}_{s}}\!n_{i}
	\end{align*}
	This completes the proof.
\end{proof}
\cref{pro:1} bridges the gap between traditional observability conditions based on the system matrices and the invertibility of the transformation matrix $\bm{T}_{i\alpha}$, which is instrumental for the following analysis.

\begin{lem}\label{rem:graph_path}
    \cite{bollobas2013modern} For a given directed graph $\mathcal{G}=\{\mathcal{N,E,A}\}$, $\mathcal{C}=[c_{ij}]\in\mathbb{R}^{N\times~N}=\mathcal{A}^{L}$ denotes the $L$th power of the adjacency matrix, then $c_{ij}$ is equal to the number of available paths from node $i$ to node $j$ in $L$ steps (across $L$ edges).
\end{lem}
From \cref{rem:graph_path}, we define $\mathcal{D}^{L}=\mathrm{bool}\left(\sum_{i=1}^{L}\mathcal{A}^{i}\right)$ the $L$-step data flow matrix with $\mathrm{bool}(\cdot)$ the boolean function, and the nonzero elements in $\mathcal{D}^{L}(:,i)$ ($i$th column of $\mathcal{D}^{L}$) indicate the nodes available to node $i$ in $L$ steps. Let $\mathcal{N}^{L}_{i} \subseteq \mathcal {N}\setminus{i},\,\forall i \in \mathcal{N}$ be the $L$-step reachable set of node $i$ with $\mathcal{N}^{0}_{i}=\emptyset$. Note that $\mathcal{N}^{L}_{i}$ can be inferred from $\mathcal{D}^{L}$ by $\mathcal{N}^{L}_{i} = \{i|\mathcal{D}^{L}(:,i) \ne 0\}$.

Without considering the network delays, each subsystem $i$ is able to acquire $\bm{T}_{j\alpha}$ and up-to-date local state estimate $\bm{\hat{z}}_{j}$ from any other sensor node $j \in \mathcal{N}^{L}_{i}$ through cross-agent communication.
Next, we introduce a definition and a necessary assumption for the solvability of \cref{prob:1}.

\begin{dfn}[Complementary  Neighbouring (CN) set]\label{dfn:GN}
		Assume $\mathcal{CN}_{i} \subseteq \mathcal{N}^{L}_{i},\,\forall\,L$, it is said to be a CN set of node $i$ if the matrix $\left[
			\begin{array}{cc}
				\bm{T}_{i\alpha}  \\
				\mathrm{col}_{j\in\mathcal{CN}_{i}}(\bm{T}_{j\alpha})
			\end{array}
			\right] \in \mathbb{R}^{(\sum{n_j}+n_i)\times n}$ is full rank, i.e.,
		\begin{align}\label{eq:rankcon}
		    \mathrm{rank}\Bigg(\left[
			\begin{array}{cc}
				\bm{T}_{i\alpha}  \\
				\mathrm{col}_{j\in\mathcal{CN}_{i}}(\bm{T}_{j\alpha})
			\end{array}
			\right]\Bigg)=n\,.
		\end{align}
\end{dfn}

\begin{ass}\label{ass:CN_existence}
Each node $i \in \mathcal{N}$ of the system \eqref{eq:sys} has at least one CN set.
\end{ass}
As it can be noticed, the common assumptions of strongly connected graphs \cite{han2019simple,silm2019note,silm:tac20} are relaxed in this paper by \cref{ass:CN_existence}, which can hold in the absence of strong connectivity. \cref{ass:CN_existence} guarantees the existence of a CN set for each node. Nevertheless, without imposing further constraints, for any agent $i$, its CN set may not be unique and redundant information might be exchanged. In sequel, we show how to find a class of optimised CN sets $\bm{\mathcal{CN}}^{opt}=\{\mathcal{CN}^{opt}_{1},\cdots,\mathcal{CN}^{opt}_{N}\}$ in terms of the communication cost for the proposed distributed observer, and how to avoid redundant data exchange.  

\begin{ass}\label{ass:communicationcost}
The graph $\mathcal{G}$ modelling the communication network of the distributed system \eqref{eq:sys} has an equal communication cost across all edges.
\end{ass}

{Consider $h^{L}_{i}$ the $L$-step neighbouring set of node $i$ required for exchanging data (possibly cross-agent) with agent $i$, $\min \left|\mathcal{CN}^{opt}_{i}\right|$ is found according to
\begin{align}
    \begin{aligned}
\{h^{1}_{i},h^{2}_{i},\cdots,h^{P}_{i}\} = \argmin_{h^{L}_{i} \subseteq \mathcal{N}_i^L}\left|\mathcal{CN}^{opt}_{i}\right|, \,\,\text{such that}\\ \mathrm{rank}\Bigg(\left[
    \begin{array}{cc}
         \bm{T}_{i\alpha}  \\
    \mathrm{col}_{j\in\mathcal{CN}_{i}^{opt}}(\bm{T}_{j\alpha})
      \end{array}
    \right]\Bigg)\!=\!n 
    \end{aligned}\label{eq:communication_opt}
\end{align}
provided $P< |\mathcal{N}|$ the minimum step value to ensure the rank condition, such that
$$\mathrm{rank}\Bigg(\!\left[\!\!
    \begin{array}{cc}
         \bm{T}_{i\alpha}  \\
    \mathrm{col}_{j\in\mathcal{N}_{i}^{P}}(\bm{T}_{j\alpha})
      \end{array}
    \!\!\right]\!\Bigg)\!=\!n \!>\!\mathrm{rank}\Bigg(\!\left[\!\!
    \begin{array}{cc}
         \bm{T}_{i\alpha}  \\
    \mathrm{col}_{j\in\mathcal{N}_{i}^{P-1}}(\bm{T}_{j\alpha})
      \end{array}
    \!\!\right]\!\Bigg),$$
It is worth noting that the optimisation problem \eqref{eq:communication_opt} does not necessarily lead to the minimum $\left|\mathcal{CN}^{opt}_{i}\right|$ in a global sense as a smaller $\left|\mathcal{CN}^{opt}_{i}\right|$ may be obtained by searching up to a step value great than $P$. However, by constraining the outreach step value at $P$, it is beneficial for mitigating the impact of cross-agent communication delay, as will be discussed later on in \cref{subsec:withdelays}.
}
By \eqref{eq:communication_opt}, we provide the offline optimisation algorithm for the selection of a CN set as summarised in Algorithm~\ref{alg:1}. It is noteworthy that under \cref{ass:communicationcost}, the solution to the optimisation problem \eqref{eq:communication_opt} may not be unique unless additional constraints are imposed. Moreover, in case that $\sum_j n_j>n-n_i,\,j\in\mathcal{CN}_{i}^{opt}$, the matrix $\left[
    \begin{array}{cc}
         \bm{T}_{i\alpha}  \\
    \mathrm{col}_{j\in\mathcal{CN}_{i}^{opt}}(\bm{T}_{j\alpha})
      \end{array}
    \right]$ has more than $n$ rows, which implies information redundancy. In this context, Algorithm \ref{alg:1} also extracts $n-n_i$ rows 
    from $\mathrm{col}_{j\in\mathcal{CN}_{i}^{opt}}(\bm{T}_{j\alpha}) \in \mathbb{R}^{\sum_jn_j \times n}$, such that the resulting matrix
    $ \left[
    \begin{array}{cc}
         \bm{T}_{i\alpha}  \\
         \mathrm{col}_{j\in\mathcal{CN}_{i}^{opt}}(\bm{T}_{j\alpha}^*)
      \end{array}
    \right] \in \mathbb{R}^{n \times n}$ with $\bm{T}_{j\alpha}^*$ the extracted row elements from $\bm{T}_{j\alpha}$, 
    is rank of $n$.

\begin{algorithm}[!ht]
    \label{alg:1}
    \LinesNumbered
	\caption{Offline optimisation of the data acquisition scheme}
	\KwIn{system matrices $\bm{A}$ and $\bm{C}$; graph adjacency matrix $\mathcal{A}$; node number $N$}
	\KwOut{optimised CN sets $\bm{\mathcal{CN}}^{opt}$}
	\BlankLine
	\textbf{Initialisation:} iteration index $k=1,l=1$\;
	\While{$\bm{\mathcal{CN}}^{opt}$ is not obtained}{
	calculate $\mathcal{D}^{l}$\;
	\For{$k\leftarrow1$ \KwTo $N$}{
	\If{$\mathcal{CN}^{opt}_{k}$ is not obtained}{
	calculate $\mathcal{N}^{l}_{k}$ based on $\mathcal{D}^{l}$\;
	optimise $\mathcal{CN}^{opt}_{k}$ using \eqref{eq:communication_opt} and identify $\bm{T}_{j\alpha}^*$\;
	}}
	$l=l+1$\;
	}
\end{algorithm}

\begin{rem}
    Algorithm~\ref{alg:1} optimises the communication network under Assumption~3, which assumes a uniform weight across all edges. A more general framework can be modelled by utilising a weighted communication graph, where each edge is assigned a weight associated with an individual communication cost. This calls for an optimisation of the information exchange architecture {to minimise} the aggregated cost from $1$-step reachable set to $P$-step reachable set rather than the cost for the $P$th step only as in the present paper (see \eqref{eq:communication_opt}), provided $P$ the minimum step number to render the full rank condition \eqref{eq:rankcon}. A detailed discussion of this subject is beyond the scope of the present article, but it is envisaged to be done in future work.
\end{rem}

 As $\bm{T}_{j\alpha}^*,\,\forall j\in\mathcal{CN}_{i}^{opt}$ is determined offline, it is known to each node $i$ when the communication network is initialised. Furthermore, the data sets required by each node $i$ in real-time for global state observation is defined as 
    \begin{equation}\label{eq:exchangeddata}
       \mathcal{I}_{ij} = \{\bm{\hat z}_j^*\},\,\forall j \in \mathcal{CN}^{opt}_{i}      
       \end{equation}
 where $\bm{\hat z}_j^*$ the local estimate of the ${\bm z}_j^*$ that fulfils $ {\bm z}_j^*= \bm{T}_{j\alpha}^*\bm{x}$. 

In view of the linear relation \eqref{eq:zi2x}, each agent $i$ can obtain the full state vector by
\begin{equation}
		\bm{\hat{x}}_{i}=\left[
    \begin{array}{cc}
         \bm{T}_{i\alpha}  \\
         \mathrm{col}_{j\in\mathcal{CN}_{i}^{opt}}(\bm{T}_{j\alpha}^*)
      \end{array}
    \right]^{-1}\left[
    \begin{array}{cc}
         \bm{\hat z}_i  \\
         \mathrm{col}_{j\in\mathcal{CN}_{i}^{opt}}(\bm{\hat z}_j^*)
      \end{array}
    \right],\forall t>0
   \label{eq:thm1}
  \end{equation}
provided the data sets $\mathcal{I}_{ij},\,\forall j\in\mathcal{CN}^{opt}_{i}$ via communication. Hence, the fixed-time convergent condition \eqref{rq:mainproblem} can be achieved. However, in practice, due to the various delays consist in the network, \eqref{eq:thm1} will not work without further provisions, which will be provided in the next subsection.

\subsection{Delayed Case}
\label{subsec:withdelays}
We now have all the ingredients to propose our main algorithm for the practical case, where network delays exist. For the sake of further analysis, let $\tau_{ij}$ be the time-varying delay consists in gathering the information set $\mathcal{I}_{ij}$ from $j$.

\begin{ass}\label{ass:boundedDelay}
For any $i \in \mathcal{N}$, the accumulated delays is bounded, such that $\sum_j\tau_{ij} \leq \overline{\tau},\,\forall j\in \mathcal{CN}^{opt}_{i}$, with $\overline{\tau}$ a known positive constant.
\end{ass}

Under \cref{ass:boundedDelay}, we assume that all the sensor nodes and observers have synchronized clocks and include time-stamps in the date transmission~\cite{kruszewski2012switched}. As such, each node $i$ can identify at $t \geq \overline{\tau} $ a set of  $\mathcal{I}_{ij}(t-\overline{\tau}),\,\forall j\in \mathcal{CN}^{opt}_{i}$ with synchronised delay. Combined with $\bm{\hat z}_i(t-\overline{\tau})$ the distributed observer can be designed, as shown in the following theorem.

\begin{thm}[Distributed Fixed-Time Observer]\label{thm:2}
Under Assumptions \ref{ass:observability}, \ref{ass:CN_existence} and \ref{ass:boundedDelay}, given the distributed system \eqref{eq:sys}, the fixed-time estimation scheme \eqref{eq:obs_eq20} and the intercommunication mechanism determined by Algorithm~\ref{alg:1}, for each node $i\in \mathcal{N}$, the local state estimate $\bm{\hat{x}}_{i}(t), \forall t \geq \overline{\tau} +t_\delta$ obtained by 
	\begin{equation}
		\bm{\hat{x}}_{i}(t)\!=\!e^{\bm{A}\overline{\tau}}\left[\!\!
    \begin{array}{cc}
         \bm{T}_{i\alpha}  \\
         \mathrm{col}_{j\in\mathcal{CN}_{i}^{opt}}(\bm{T}_{j\alpha}^*)
      \end{array}
    \!\!\right]^{-1}
    \!\left[\!\!
    \begin{array}{cc}
        \bm{\hat z}_i(t-\overline{\tau})  \\
         \mathrm{col}_{j\in\mathcal{CN}_{i}^{opt}}(\bm{\hat z}_j^*(t-\overline{\tau}))
      \end{array}
    \!\!\right]
    \label{eq:thm2}
  \end{equation}
for all $t \geq \overline{\tau} +t_\delta$ is equal to ${\bm x}(t)$, such that the condition \eqref{rq:mainproblem} is fulfilled.
\end{thm}
\begin{proof}
Thanks to the finite-time local observers \eqref{eq:obs_eq20} that is activated at $t=t_\delta$ and the information sets $\mathcal{I}_{ij}$ received from the CN set, node $i$ is able to reconstruct delayed estimates
\begin{multline*}
    \left[
    \begin{array}{cc}
         \bm{\hat z}_i(t-\overline{\tau})  \\
         \mathrm{col}_{j\in\mathcal{CN}_{i}^{opt}}(\bm{\hat z}_j^*(t-\overline{\tau}))
      \end{array}
    \right] \!=\!  \left[
    \begin{array}{cc}
         {\bm z}_i(t-\overline{\tau})  \\
         \mathrm{col}_{j\in\mathcal{CN}_{i}^{opt}}({\bm z}_j^*(t-\overline{\tau}))
      \end{array}
    \right].
\end{multline*}
for all $t \geq \overline{\tau} +t_\delta$ is equal to ${\bm x}(t)$. Hence, from \eqref{eq:thm1} and \eqref{eq:sys}, it is immediate to show the following relationship by using \eqref{eq:thm2}
\begin{align}
   	\bm{\hat{x}}_{i}(t)=e^{\bm{A}\overline{\tau}}{\bm x}(t-\overline{\tau}) = \bm{x}(t),\,\forall t \geq \overline{\tau} + t_\delta \label{eq:thm1_proof}
\end{align}
which completes the proof.
\end{proof}

\begin{rem}\label{rem:1}
In contrast to the existing distributed observers (e.g., Luenberger-like observers) where agents only communicate with their neighbours (i.e., $\mathcal{N}_i$), the proposed method relies on a cross-agent communication strategy which enables an agent $i$ to communicate with $j \notin \mathcal{N}_i$. This feature enables the optimisation Algorithm~\ref{alg:1}, and the resulting data flow may turn out to be efficient and useful in practice to reduce communication load. More specifically, in the proposed distributed observer, the accumulated data flows into a node is of dimension $(n-n_{i})$, thus the dimension of the data flow through a communication channel (i.e., an edge, in one direction) is below $(n-n_{i})$. However, in Luenberger-like distributed observers~\cite{han2019simple}, the data transmitted along any edge is of dimension $n$, and each node has to manage to collect $n|\mathcal{N}_{i}|$-dimensional data. Despite the delay introduced by the cross-agent communication, {the influence of a bounded delay can be compensated using an open-loop prediction scheme.
}
\end{rem}




\begin{rem}
With the proposed cross-agent communication strategy, the proposed estimation scheme remains valid if the outputs $y_i$ are shared instead of the local state estimates $\bm{\hat z}_j$ (see \cref{pro:1} that builds the connection between conventional observability and the invertibility of the coordinate transformation from $\bm{x}$ to ${\bm z}$). Nevertheless, 
sharing the outputs directly may sacrifice \textit{privacy-preserving} properties of the method. Particularly, when one or more sensors/communication links are attacked, it could expose more sensor nodes to the attacker. For this reason, the state estimate sharing strategy is adopted in the proposed framework. 
Moreover, the cross-agent communication strategy may be applied to either Luenberger-type asymptotic \cite{han2019simple} or finite-time \cite{silm2019note} observers, which can also leads to reduced communication as discussed in \cref{rem:1}.
\end{rem}


\section{Robustness Analysis of the Observer}
\label{sec:4}
This section analyses the robustness of the proposed observer against measurement and process disturbances.
Assuming the presence of the bounded model uncertainty and sensor disturbance,  $\|\bm{d}_{x}\|_{\infty}\leq\overline{d}_{x},\|\bm{d}_{y,i}\|_{\infty}\leq\overline{d}_{y}$ in \eqref{eq:sys}, such that
\begin{align}\label{eq:noisy_sys}
    \bm{\dot{x}}_d = \bm{Ax}_d+\bm{d}_{x},\ \bm{y}_d = \bm{Cx}_d+\bm{d}_{y}
\end{align}
where $\bm{x}_d$ denotes state variable under the effect of $\bm d_x(t)$.
In this context, for $i$th subsystem it holds that 
$ 
    \bm{\dot{\bar x}} _{i,d}= \bm{\bar{A}}_i \bm{\bar  x}_{i,d}+\bm T_i\bm{d}_{x}, {y}_{i,d} = \bm{\bar{ C}}_i\bm{ \bar x}_{i,d}+{d}_{y,i}
$
where $\bm{{\bar x}}_{i,d}=[\bm{\bar x}_{io,d}^\top\quad \bm{\bar x}_{iu,d}^\top]^\top $. 
By analogy to \eqref{eq:obs_canonical}, the observable part follows
\begin{align}\label{eq:noisy_sysz} 
    \bm{\dot{ z}}_{i,d} = \bm{A}_{i,z}   \bm z_{i,d}+\bm{{d}}_{z,i}, \bm{y}_{i,d} = \bm{ C}_{i,z}  \bm z_{i,d}+{d}_{y,i},
\end{align}
where { $\bm{d}_{z,i}\triangleq {\bm T}_{i\alpha}\bm{d}_{x}=[d_{z,i,0},\cdots,d_{z,i,n_i-1}]^\top\in\mathbb{R}^{n_i}$ } and the disturbance-effected state variable satisfies the identity that
\begin{align}\label{eq:noisy_z}
\bm\gamma_{i,h}\bm z_{i,d}=\lambda_{i,h,d}
\end{align}
with 
\begin{align*}
    \lambda_{i,h,d}\!=&(-1)^{n_{i}-1}\!\!\left[\!V_{K_{i,h}^{(n_{i})}}\bm C_{i,z}\bm z_{i,d}\!\right]\!\!+\!\!\!\sum_{p=0}^{n_{i}-1}\!\!a_{p}(-1)^{p}\!\!\left[\!V_{K_{i,h}^{(p)}}\bm C_{i,z}\bm z_{i,d}\!\right]\\
&	+\sum_{p=0}^{n_{i}-1}(-1)^{p}\left[V_{K_{i,h}^{(p)}}{ d_{z,i,n_{i}-1-p}}\right].
\end{align*}
In the noisy environment, the state estimator \eqref{eq:obs_eq20} gives
\begin{align}\label{eq:noisy_obs_eq20}
	\bm{\hat{z}}_{i,d}=\bm{\Gamma}_i^{-1}\bm{\hat \Lambda}_{i,d},\forall t \geq t_\delta
\end{align}
where $\bm{\hat \Lambda}_{i,d}=\left[\hat \lambda_{i,0,d},\hat \lambda_{i,1,d},\cdots,\hat \lambda_{i,n_{i}-1,d}\right]^\top$ and 
\begin{align*}
    \hat \lambda_{i,h,d}=&(-1)^{n_{i}-1}\left[V_{K_{i,h}^{(n_{i})}}y_{i,d}\right]+\sum_{p=0}^{n_{i}-1}a_{p}(-1)^{p}\left[V_{K_{i,h}^{(p)}}y_{i,d}\right].
\end{align*}
Comparing \eqref{eq:noisy_z} and \eqref{eq:noisy_obs_eq20}, the estimation error of $\bm z_{i,d}$ takes on the form
\begin{align}\label{eq:noisy_sysz_error} 
    \bm{\epsilon} _{z,i}\triangleq \bm z_{i,h,d}-\bm{\hat z}_{i,h,d}=\bm \Gamma_i^{-1}\bm \epsilon_{\Lambda,i}
\end{align}
where $\bm\epsilon_{\Lambda,i}=\left[\epsilon_{\lambda_i,0},\epsilon_{\lambda_i,1},\cdots,\epsilon_{\lambda_i,n_{i}-1}\right]^\top$, and 
\begin{equation*}
\begin{array}{cll}
    \epsilon_{\lambda_i,h}&\triangleq&  \lambda_{i,h}-\hat \lambda_{i,h}
    =	\displaystyle\sum_{p=0}^{n_{i}-1}(-1)^{p}\left[V_{K_{i,h}^{(p)}} d_{z,i,n_{i}-1-p}\right]\\
	&&	-(-1)^{n_{i}-1}\!\!\left[V_{K_{i,h}^{(n_{i})}}d_{y,i}\right]\!\!-\!\!\displaystyle\sum_{p=0}^{n_{i}-1}a_{p}(-1)^{p}\!\!\left[V_{K_{i,h}^{(p)}}d_{y,i}\right].
	\end{array}
\end{equation*}
The effects of both measurement noise $d_{y,i}$ and model uncertainty $\bm{d}_{x}$ are embedded in $\bm{\epsilon}_{{\Lambda},i}$ in the form of Volterra images, i.e. $\left[V_{K_{i,h}^{(p)}}d_{y,i}\right]\triangleq\epsilon_{d_{y,i},p,h}$ and $\left[V_{K_{i,h}^{(p)}}d_{z,i,p}\right]\triangleq\epsilon_{d_{z,i},p,h}, p\in\{0,\cdots,n_{i}\},h\in\{0,\cdots,n_{i}-1\}$.
Recall the transformation of the Volterra operator, $\epsilon_{d_{y},p,h}$ is the output of the LTV system
\begin{align}\label{eq:V_Kdy}
	\dot{\epsilon}_{d_{y},p,h}=-\omega_{h}\epsilon_{d_{y},p,h}+K_{i,h}^{(p)}(t,t)d_{y}
\end{align}
Thanks to the Bounded-Input-Bounded-Output(BIBO) feature of \eqref{eq:V_Kdy}, effects of the measurement noise can be bounded by
\begin{align*}
	|\epsilon_{d_{y,i},p,h}|\leq\left|\frac{1}{\omega_{h}}\overline{d}_{y,i}\sup_{0<\tau\leq t}K_{i,h}^{(p)}(\tau,\tau)\right|\triangleq\overline{\epsilon}_{d_{y,i},p,h}
\end{align*}

In the same line of reasoning, the Volterra images of the model uncertainty $\bm{d}_{z,i}$ have the upper bound $\overline{\epsilon}_{d_{z,i},p,h}\triangleq{\left\|\bm T_{i\alpha}\right\|_\infty}\left|\frac{1}{\omega_{h}}\overline{d}_{x,i}\sup_{0<\tau\leq t}K_{i,h}^{(p)}(\tau,\tau)\right|$. Therefore, the overall upper bound of the state estimation error $\overline{\epsilon}_{\lambda_{i},h}$ for all $h\in\{0,1,\cdots,n_{i}-1\}$ can be written as
\begin{align*}
	|\epsilon_{\lambda_{i},h}|\leq\overline{\epsilon}_{d_{y,i},n,h}+\sum_{p=0}^{n_{i}-1}|a_{p}|\overline{\epsilon}_{d_{y,i},p,h}+\sum_{p=0}^{n_{i}-1}\overline{\epsilon}_{d_{z,i},p,h}\triangleq\overline{\epsilon}_{\lambda_{i},h}.
\end{align*}

As such, by stacking $\overline{\epsilon}_{\lambda_{i},h}$ induced by different kernels, one can obtain the vector bound as $\bm{\overline{\epsilon}}_{\Lambda,i}\triangleq\left[\overline{\epsilon}_{\lambda_{i},0},\cdots,\overline{\epsilon}_{\lambda_{i},n_i-1}\right]^{\top}$. Consequently, the observation error defined in \eqref{eq:noisy_sysz_error} is bounded by $
   \bm{\overline{\epsilon}}_{{z}_{i}}\leq\|\bm{\Gamma}_i^{-1}\|_\infty\bm{\overline{\epsilon}}_{{\Lambda},i}$.

Taking the communication delay into account, the compensation in \eqref{eq:thm2} writes
\begin{equation}\label{eq:hatx_disturbance}
    	\bm{\hat{x}}_{i,d}(t)\!=\!e^{\bm{A}\overline{\tau}}\left[\!\!\!
    \begin{array}{cc}
         \bm{T}_{i\alpha}  \\
         \mathrm{col}_{j\in\mathcal{CN}_{i}^{opt}}(\bm{T}_{j\alpha}^*)
      \end{array}
    \!\!\!\right]^{-1}\!
    \left[\!\!\!
    \begin{array}{cc}
         \bm{\hat z}_{i,d}(t-\overline{\tau})  \\
         \mathrm{col}_{j\in\mathcal{CN}_{i}^{opt}}(\bm{\hat z}_{j,d}^*(t-\overline{\tau}))
      \end{array}
   \! \!\!\right]
\end{equation}
However, recalling \eqref{eq:noisy_sys}, during the delay $\overline\tau$, $\bm d_x$ introduce extra effects that can be expressed as
\begin{align}
   \bm\epsilon_{dx,\overline\tau}= \displaystyle\int_{t-\overline\tau}^te^{\bm A(t-\tau)}\bm{d}_{x}(\tau)d\tau
\end{align}
Being $\bm A$ Hurwitz, it is straightforward to conclude that $ \epsilon_{dx,\overline\tau}$ is bounded with an upper bound $\overline \epsilon_{dx,\overline\tau}\geq\left\|\displaystyle\int_{t-\overline\tau}^t e^{\bm A(t-\tau)}{\bm d}_x(\tau)d\tau\right\|_{\infty}$.

Notably, for any $i\in\mathcal{N}$ in \eqref{eq:thm1}, $\bm{\Gamma}_{i}$ and $\bm{T}_{i\alpha}$ are not affected by model uncertainty and disturbances. Therefore, the distributed observation of $\bm x_{i,d}$ in \eqref{eq:hatx_disturbance} remains bounded as long as $\bm{d}_{x},{d}_{y}$ are bounded, i.e., for all $t \geq \overline{\tau} +t_\delta$,
\begin{align}
	\begin{aligned}
		&\left|\bm{\epsilon}_{{\hat{x}}_{i}}\right|\leq\left[\!\begin{array}{cc}
         \bm{T}_{i\alpha}  \\
         \mathrm{col}_{j\in\mathcal{CN}_{i}^{opt}}(\bm{T}_{j\alpha}^*)
      \end{array}\!\right]^{-1}\!
      \left[\!
        \begin{array}{cc}
            \bm{\overline{\epsilon}}_{{z}_{i}}  \\
            \mathrm{col}_{j\in\mathcal{CN}_{i}^{opt}}(\bm{\overline{\epsilon}}_{{z}^{*}_{j}})
        \end{array}
       \! \!\right]+\overline \epsilon_{dx,\overline\tau},
	\end{aligned}
	\label{eq:robust_noise}
\end{align}

\section{Numerical Examples}
\label{sec:5}
In this section, the effectiveness of the proposed distributed observer is examined by a few numerical examples. Consider a linear system \cite{han2019simple} of order $n=6$ with four local sensors, i.e., $N=4$ and $[n_1\ n_2\ n_3\ n_4]=[2\ 5\ 1\ 5]$. System parameters are given as follows
\begin{align*}
    \bm{A}=\begin{bmatrix}
    -1 & 0 & 0 & 0 & 0 & 0 \\
    -1 & 1 & 1 & 0 & 0 & 0 \\
    1 & -2 & -1 & -1 & 1 & 1 \\
    0 & 0 & 0 & -1 & 0 & 0 \\
    -8 & 1 & -1 & -1 & -2 & 0 \\
    4 & -0.5 & 0.5 & 0 & 0 & -4 \\
    \end{bmatrix}\\
    \bm{C}=\left[\begin{array}{cccccc}
    1 & 0 & 0 & 2 & 0 & 0 \\
    2 & 0 & 0 & 1 & 0 & 0 \\
    \hdashline
    2 & 0 & 1 & 0 & 0 & 1 \\
    \hdashline
    0 & 0 & 0 & 2 & 0 & 0 \\
    \hdashline
    1 & 0 & 2 & 0 & 0 & 0 \\
    2 & 0 & 4 & 0 & 0 & 0 \\
    \end{array}\right]=\begin{bmatrix}
    \bm{C}_{1} \\ \bm{C}_{2} \\ \bm{C}_{3} \\ \bm{C}_{4}
    \end{bmatrix}
\end{align*}
with the communication network described by
\begin{align}\label{eq:exampleA}
	\mathcal{A}=\begin{bmatrix}
		0 & 0 & 1 & 1\\
		1 & 0 & 1 & 0 \\
		0 & 1 & 0 & 0 \\
		1 & 0 & 0 & 0 \\
	\end{bmatrix}\,.
\end{align}
\begin{figure}[!htb]
	\centering
	\includegraphics[width=\columnwidth]{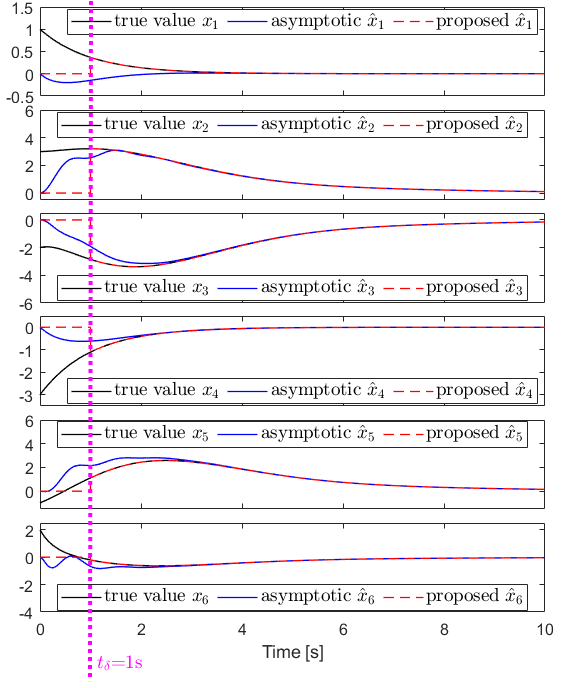}\\[-2.5ex]
	\caption{State estimates of method \cite{han2019simple} and the proposed method in the delay-free and noise-free scenario.}
	\label{fig:state_delay_free_noise_free}
\end{figure}
From \eqref{eq:exampleA} and Algorithm~\ref{alg:1}, it is straightforward to obtain $\mathcal{CN}^{opt}_{1}=\{2\},\,\mathcal{CN}^{opt}_{2}=\{3\},\,\mathcal{CN}^{opt}_{3}=\{2\},\,\mathcal{CN}^{opt}_{4}=\{1\}$.
It is noteworthy that the redundant communication links which are $\{4\}$ in $\mathcal{CN}_{1}$ and $\{1\}$ in $\mathcal{CN}_{2}$ are removed at the design stage by Algorithm~\ref{alg:1}, thereby reducing the data transfer required by a distributed observer. Moreover, taking the subsystem $1$ as an example, it does not require the full observable states of the subsystem $2$ owing to the information redundancy $n_1+n_2>n$ as discussed in \cref{subsec:withoutdelay}.

In the first instance, the communication delay is neglected throughout the network, and both process and measurement noises are not taken into account. The simulation results show that the state estimates of all agents can reach consensus immediately after the activation time $t_\delta=1$s. For benchmarking purposes, the estimation results of the proposed method is compared with a recently proposed Luenberger-like  approach~\cite{han2019simple}.
Taking the $1$st subsystem as an example, the comparative results are plotted in \cref{fig:state_delay_free_noise_free}. As it can be seen, all estimates of the proposed method converge to the actual state within a fixed time, showing a much faster convergence speed than the method in \cite{han2019simple}.
\vspace{-2ex}
\begin{figure}[!htb]
	\centering
	\includegraphics[width=\columnwidth]{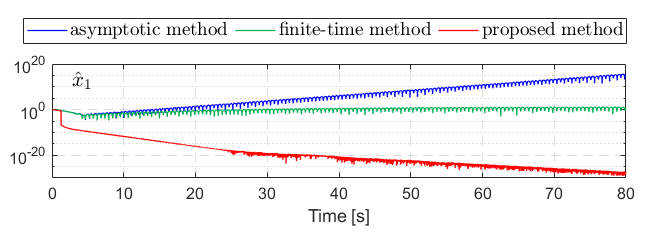}\\[-2.5ex]
	\caption{State estimation errors of methods \cite{han2019simple,silm2019note} and the proposed method in the delayed and noise-free scenario.}
	\label{fig:error_delayed_noise_free}
\end{figure}
Next, a uniform delay is added to each network edge with an upper bound $\overline{\tau}=0.27$s. In addition to \cite{han2019simple}, we further compare the proposed method with a finite-time distributed observer that has been shown robust against communication delays~\cite{silm2019note}.
The errors of the state estimates are given in \cref{fig:error_delayed_noise_free}. Under the delayed network, { the asymptotic method~\cite{han2019simple}} shows the non-convergent performance, while the error of the finite-time observer~\cite{silm2019note} stays bounded. However, the proposed method demonstrates its advantage in terms of dealing with network delays by showing the most accurate state estimation.
\begin{figure}[!htb]
	\centering
	\includegraphics[width=\columnwidth]{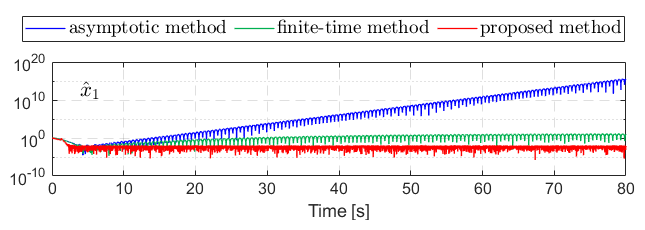}\\[-2.5ex]
	\caption{State estimation errors of methods \cite{han2019simple,silm2019note} and the proposed method in the delayed and noisy scenario.}
	\label{fig:error_delayed_noisy}
\end{figure}

\begin{figure}[!htb]
	\centering
	\includegraphics[width=\columnwidth]{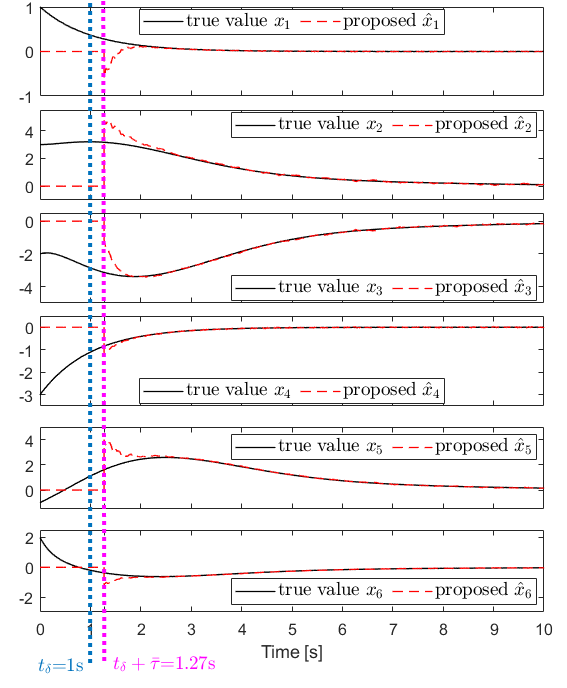}\\[-2.8ex]
	\caption{State estimates of the proposed method in the delayed and noisy scenario.}
	\label{fig:state_delayed_noisy}
\end{figure}

Finally, a noisy scenario is simulated where the outputs are corrupted by a uniformly distributed random noise within $\left[-0.2\quad0.2\right]$ and the system dynamic are perturbed by a sinusoidal uncertainty $0.1\sin(50t)$. Under the effects of both disturbances and the same delay considered in the previous example, the estimation error of all three methods are compared in \cref{fig:error_delayed_noisy}, where the proposed method 
outperforms the other two in terms of steady-state accuracy. From the state estimates shown in \cref{fig:state_delayed_noisy}, the proposed method converges within a fixed time $t_{\delta}+\bar{\tau}=1.27$s. This arises from that once the proposed observer is activated at $t_{\delta}=1$s, it requires at most $\bar{\tau}$ to transmitting neighbouring information ensuring the fully observable in each subsystem.

\section{Conclusion}
\label{sec:6}
A fixed-time convergent observer is proposed for distributed state estimation of a large scale system with directed communication typologies. The fast convergence properties enable the data transmission delay to be compensated a posteriori. As such, cross-agent communication is utilised and it yields a more effective data exchange mechanism with an optimised (minimised) data flow. The boundedness of the estimation error has been confirmed subject to bounded measurement and process disturbances. Numerical examples and comparisons with the existing method have been shown to verify the effectiveness of the proposed method. Future research efforts will be devoted to including a time-varying communication graph, event-triggered communication, as well as to consider more general nonlinear systems and cyber security issues.


%





\bibliographystyle{IEEEtran}
\bibliography{reference}

\end{document}